\documentclass[submission,copyright,creativecommons]{eptcs}
\providecommand{\event}{LTT 2026}
\usepackage{amsmath,amssymb,amsthm,hyperref}

\def\RR{{\mathbb R}}
\def\NN{{\mathbb N}}

\usepackage{iftex}

\ifpdf
  \usepackage{underscore}         
  \usepackage[T1]{fontenc}        
\else
  \usepackage{breakurl}           
\fi

\newcommand{\ba}{\begin{array}} \newcommand{\ea}{\end{array}}
\newfont{\bsl}{cmbxsl10 scaled 1095}

\newfont{\deu}{eufm10 scaled 1000}

\newcommand{\res}{|\!\raisebox{1mm}{$\scriptscriptstyle\setminus$}}

\newcommand{\aquant}{\forall}
\newcommand{\equant}{\exists}

\newcommand{\beq}{\begin{equation}}
\newcommand{\eeq}{\end{equation}}

\begin{document}

\title{On the Computational Content of
Moduli of Regularity and their Logical
Strength}

\author{Ulrich Kohlenbach\\ 
Department of Mathematics \\  Technische Universit\" at Darmstadt\\ 
Schlossgartenstra\ss{}e 7, 64289 Darmstadt, Germany \\ 
kohlenbach@mathematik.tu-darmstadt.de \\[4mm] Dedicated to Professor Stefano Berardi on the occasion of his 64th Birthday}

\newcommand{\titlerunning}{On the computational content of
moduli of regularity and their logical
strength}
\newcommand{\authorrunning}{Ulrich Kohlenbach}

\maketitle 
\begin{abstract} 
  We continue the investigation into the computational status
  of the existence of moduli of regularity (and their use for rates of
  convergence) in the sense of
  \cite{KohlenbachLopezNicolae(2019)},
  carried out w.r.t. classical reverse mathematics
  and Weihrauch degrees in \cite{Kohlenbach(moduli)},
  and determine the amount of LEM involved.
  We also show that the existence of a modulus of regularity
  always yields an algorithm for the computation of a zero in the case
  of continuous functions $F:K\to\RR$ on a compact metric space (in $F$
  equipped with a modulus of uniform continuity and $K$ given in standard
  representation) whenever such a zero exists. If $K\subset X$ is a compact
  subset of a uniformly convex Banach space $X$ and the zero set
  of $F$ is convex one can compute even the zero of minimal norm.
  A modulus of regularity can also be used to compute
  the left-most infinite path of an infinite 0/1-tree. We also show that
  there is no proof-theoretically tame nonstandard uniformity principle which
  would make it possible to replace in the regularity assumption
  compactness by metric boundedness and still guarantee
  classically correct bounds.
\end{abstract}
\newtheorem{definition}{Definition}[section]
\newtheorem{proposition}[definition]{Proposition}
\newtheorem{remark}[definition]{Remark}
\newtheorem{theorem}[definition]{Theorem}
\newtheorem{corollary}[definition]{Corollary}
\newtheorem{lemma}[definition]{Lemma}
\newtheorem{exercise}[definition]{Exercise}
\newtheorem{clm}[definition]{Claim}
\newtheorem{prop}[definition]{Proposition}
\newtheorem{example}[definition]{Example}
\newtheorem{notation}[definition]{Notation}
\newtheorem{application}[definition]{Application} 
\section{Introduction}
In this paper, we continue our investigation from \cite{Kohlenbach(moduli)}
into the computational and logical
strength of (the existence of) of the general form of (metric)
regularity introduced in \cite{KohlenbachLopezNicolae(2019)} as a unifying
concept of many related notions studied in continuous optimization. Roughly
speaking, the assumption of regularity of a solution set allows one to
conclude that a sufficiently good approximate solution must be close to an
actual solution. If the solution set is a singleton set
this has been studied under
the name of strong or uniform uniqueness (also with moduli) e.g. in
\cite{Kohlenbach(93)} (see also \cite{Kohlenbach(book)}).
The concept of regularity generalizes this to the
nonunique case. In the context of compact metric spaces $X$ and
continuous functions $F:X\to \RR$ the regularity of the set $\text{zer}\;F$ of
zeros of $F$ always holds by a result in \cite{KohlenbachLopezNicolae(2019)}.
However, as shown in \cite{Kohlenbach(moduli)}, one usually - in contrast
to the case of uniqueness where proof-theoretic techniques can be used to
extract an effective modulus of uniqueness - cannot hope for a computable
modulus of regularity. In fact, the existence of the latter is - even for
Lipschitz continuous functions $F:[0,1]\to \RR$ - equivalent to arithmetical comprehension ACA$_0$
while the $\forall\exists$-form of regularity (without a modulus) only
requires WKL$_0.$ The reason for this difference is that the
$\forall\exists$-form of regularity (while proof-theoretic being weak)
implies intuitionistically (and actually is equivalent to)
the $\Sigma^0_1$-law-of-excluded-middle principle
$\Sigma^0_1$-LEM: see Theorem \ref{regularity-LEM} below.\footnote{A
  particularly striking example of such a situation is given by Ramsey's
  theorem for pairs and two colors which - though proof-theoretically weak -
  implies (and is equivalent to) even the principle $\Sigma^0_3$-LLPO
  (see \cite{BerardiSteila(14)}) which by \cite{Akama} is strictly stronger
  than $\Sigma^0_2$-LEM. See also \cite{BerardiSteila(17)} for a
generalization to $k$-many colors.}   
When strengthened into a modulus,
this use of $\Sigma^0_1$-LEM then becomes even classically visible in the
form of  $\Sigma^0_1$-comprehension (and so - by iteration - as ACA$_0$).
\\ In \cite{KohlenbachLopezNicolae(2019)}
it is shown that a modulus of regularity
yields a rate of convergence whenever we have a Fej\'er monotone (w.r.t.
the solution set) algorithm for the computation of approximate solutions
together with an approximate solution bound for this algorithm.\footnote{For
a recent extension of this result to a generalized form of Fej\'er
monotonicity see \cite{KohlenbachPinto(Fejer)}.}
  In this paper
we show that in the situation above with a compact metric space $X$ one
unconditionally can construct a primitive recursive functional which
computes uniformly in a modulus of regularity, a standard representation of
$X$ and a name for $F$ (given by the restriction of $F$ to a countable dense
subset and a modulus of uniform continuity) a zero of $F$ provided that
$\text{zer}\;F\not=\emptyset$
(Theorem \ref{compact-metric-case}). \\[1mm]
As a special case one can subsume the problem of finding an infinite path
of an infinite binary (i.e. $0/1$-)tree. Here there is a Kalmar elementary
functional which uniformly in (the characteristic function of) such a tree
and a modulus of regularity (w.r.t. the set of infinite paths as solution set)
computes the leftmost branch of the tree (Theorem \ref{leftmostbranch}). 
\\[1mm] If $K\subset X$ is a compact subset
of a uniformly convex Banach
space $X,$ $F:K\to\RR$ is continuous and $\text{zer}\;F$
is convex (a situation which frequently occurs in convex optimization) then one can compute in the above
data augmented with a modulus of convexity of $X$ even a zero of minimal
norm (Theorem \ref{uniformly-convex-case}). \\[1mm] 
In proof mining one often can allow the use of nonstandard arguments
which replace a compactness assumption by  metric boundedness. The
uniform boundedness principle $\exists$-UB$^X$ introduced in
\cite{Kohlenbach(06)} (see also \cite{Kohlenbach(book)} and the connected
discussion in \cite{FerLeuPin})
systematically makes this possible
and can - though classically being false - be eliminated without any
complexity contribution from the verification of the bounds
extracted from proofs which make use of this principle. This raises the
question whether some combination of, say, arithmetical comprehension
(which is an admissible principle in the logical metatheorems on proof
mining if one allows for bar recursive bounds) with $\exists$-UB$^X$
(or some other `tame' nonstandard principle) implies regularity even
in the absence of compactness (here a principle is called `tame' if
it does not contribute to the complexity of extractable bounds).
In Proposition \ref{counterexample}
we show that this is not the case.
\section{Main Results}
\begin{definition}[\cite{KohlenbachLopezNicolae(2019)}]\label{regularity-zero}
Let $(X,d)$ be a metric space and let be $F:X\to \RR$ a mapping. Let 
$\text{zer}\; F:=\{ x\in X: F(x)=0\}\not=\emptyset$ and $r>0.$  
We say that $F$ is regular w.r.t. $\text{zer}\; F$ and $\overline{B}(z,r)$ for 
$z\in \text{zer}\; F$ if 
\[ \forall n \in\NN\,\exists k \in \NN \, \forall x\in \overline{B}(z,r) 
  \,\left( |F(x)|<2^{-k} \to\, \exists z'\in \text{zer}\;
    F \,(d(x,z')<2^{-n})\right). \]
If this holds with `$\forall x\in \overline{B}(z,r)$' replaced by 
`$\forall x\in X$' we say that $F$ is regular w.r.t. $\text{zer}\; F.$
\\ A function $\rho:\NN\to \NN$ providing given $n$ a number
$k=\rho(n)$ satisfying 
the above is called a modulus of regularity of $F$ w.r.t. $\text{zer}\; F$ and 
$\overline{B}(z,r)$ resp. w.r.t. $\text{zer}\; F$
(short: $\rho$ mreg $\text{zer}\;F,r$)
\end{definition}
Recall that a metric space is called `proper' (or `boundedly compact') if
every bounded and closed subset is compact.
\begin{proposition}[\cite{KohlenbachLopezNicolae(2019)}]\label{existence-of-modulus}
If $X$ is proper and $F$ is continuous, then for any $z \in \text{zer}\; F$ and $r > 0$, $F$ has a modulus of regularity w.r.t.  $\text{zer}\; F$ and $\overline{B}(z,r)$.
\end{proposition}
Let $(X,d)$ be compact metric space and $(a_n)$ be a sequence in $X$ and
$\alpha:\NN\to \NN$ both together witnessing the total boundedness of
$X,$ i.e.
\[ \forall x\in X\,\forall k\in\NN\,\exists \,0\le i\le
\alpha(k)\,\left(d(x,a_i)< 2^{-k} \right)\]
(compare $(TOT I)$ in \cite{KohlenbachLeusteanNicolae}). \\ 
Let $g:\NN^2\to\NN^{\NN}$ be such that for all $i,j\in\NN,$
$g(i,j)$ is a name (in the sense of \cite{Kohlenbach(book)}) of
$d(a_i,a_j).$ \\
Let $F:X\to\RR$ be a continuous function with a modulus $\omega:\NN\to\NN$
of uniform continuity, i.e.
\[ \forall k\in\NN\,\forall x,y\in X\,\left( d(x,y)<2^{-\omega(k)}\to
    d(F(x),F(y))<2^{-k}\right).\]
Let $h:\NN\to\NN^{\NN}$ be such that for each $i\in\NN,$ $h(i)$ is a name of
$F(a_i).$
\begin{theorem}\label{compact-metric-case}
  Let $(X,d),F$ be as above. One can define a primitive recursive functional
  (in the sense of Kleene's S1-S8 from {\rm \cite{Kleene(59)}}) $\Phi$ such that for all functions
  $\omega,\alpha,g,h,\rho,$ it holds for $\beta:=\Phi(\omega,\alpha,g,h,\rho):\NN\to\NN$ that $(a_{\beta(k)})_{k\in\NN}$ converges with rate $2^{-k}$ to a zero of $F$
  provided that $\text{zer}\;F\not=\emptyset,$ the functions
  $\omega,\alpha,g,h$ satisfy the above requirements and $\rho$ is
  a modulus of regularity $\rho$ for
  $\text{zer}\;F.$
\end{theorem}
{\bf Proof:} We show how to compute primitive recursively in
$\omega,\alpha,g,h,\rho$ satisfying the above requirements a function
$\beta:\NN\to\NN$ such that for $x_k:=a_{\beta(k)}$ 
\[ (*)\ \forall k\in\NN\,\left( |F(x_k)|<2^{-\max\{ k,\rho(k+2)\}}
    \wedge (k>0\rightarrow d(x_k,x_{k-1}) \le 2^{-k})\right). \]
$(*)$ clearly implies that $(x_k)$ is a Cauchy sequence with rate $2^{-k}$ since
for $m\ge n\ge k$
\[ d(x_m,x_n)\le\sum\limits^{m-1}_{k=n} d(x_k,x_{k+1}) \le
  \sum\limits^{m-1}_{k=n} 2^{-k-1}\le \sum\limits^{\infty}_{k=n} 2^{-k-1}
  \le 2^{-n}. \]
Since $F$ is continuous, $(*)$ - moreover - implies that $x:=\lim x_k$ is
a zero of $F.$ \\
We prove $(*)$ by induction on $k$ and simultaneously define $\beta:$
Let $k=0.$ Since $\text{zer}\;F\not=\emptyset,$ one can search primitive
recursively in
$\omega,\alpha,g,h,\rho$ for an $n$ with
\[ |F(a_n)|<2^{-\rho(2)} \]
so that $(*)$ holds with $\beta(0):=n$ (for definiteness we may stipulate
that $x_{-1}:=a_0$). Indeed, let $z\in \text{zer}\;F.$ There there exists an
$n\le \alpha(\omega(\rho(2)+2))$ such that
\[ d(a_n,z) < 2^{-\omega(\rho(2)+2)}. \]
Then in turn
\[ |F(a_n)|<2^{-\rho(2)-2} \] and so
for the $2^{-\rho(2)-2}$-rational approximation
$(F(a_n))(\rho(2) +2)$ of
$F(a_n)$ provided by $\widehat{h(n)}(\rho(2)+2)$ 
\[ (1) \ |(F(a_n))(\rho(2)+2)|<2^{-\rho(2)-1}. \]  
Since $(1)$ is a primitive recursively
decidable property, one can search primitive recursively
for the least $n\le \alpha(\omega(\rho(2)+2))$ such
that $(1)$ holds which in turn implies
\[ (2) \ | F(a_n)| <2^{-\rho(2)}. \]
$k\mapsto k+1, k\ge 0:$ By the induction hypothesis we have
defined already $\beta(0),\ldots,\beta(k)$ such that
\[ (3) \ k>0\rightarrow d(x_k,x_{k-1}) \le 2^{-k} \]
and
\[ (4) \ |F(x_k)|<2^{-\rho(k+2)}. \]
By $(4)$ we get
\[ (5)\  \exists z\in X\, \left( F(z)=0\wedge d(z,x_k) < 2^{-k-2}\right). \]
Similarly to the above argument
we can search primitive recursively in $\omega,\alpha,g,h,\rho,k$ and $\beta(k)$
for an $n_k$ such that
\[ d(a_{n_k},x_k) \le 2^{-k-1}\wedge | F(a_{n_k})|< 2^{-\max\{
    k+1,\rho(k+3)\}}. \]
Indeed: let $N:=\max\{ k+1,\rho(k+3)\}.$ Let $z$ be as in $(5).$
There exists an \\ $n_k\le
\alpha\left(\max\{ k+4,\omega(N+2)\}\right)$ such that
\[ d(a_{n_k},z) < 2^{-\max\{ k+4,\omega(N+2)\} }\] and so
\[ d(a_{n_k},x_k) <2^{-k-4}+2^{-k-2} \wedge | F(a_{n_k})|<2^{-N-2}. \]
Hence
\[ (6) \  (d(a_{n_k},x_k))(k+4) < 2^{-k-3}+2^{-k-2}\wedge
(|F(a_{n_k})|)(N+2)<2^{-N-1},  \]
where $(d(a_{n_k},x_k))(k+4)$ and $(F(a_{n_k}))(N+2)$ are the $2^{-k-4}$- and
$2^{-N-2}$- rational approximations of $d(a_{n_k},x_k)$
and $F(a_{n_k})$ provided by $\widehat{g(n_k,\beta(k))}(k+4)$ and
$\widehat{h(n_k)}(N+2)$ respectively.
As $(6)$ is a primitive recursively decidable property we can search
for the least $n_k\le \alpha\left(\max\{ k+4,\omega(N+2)\}\right)$
such that $(6)$ holds
which implies that 
\[ (7) \  d(a_{n_k},x_k) <2^{-k-1} \wedge | F(a_{n_k})|<2^{-N} .\]
Now take $\beta(k+1):=n_k.$ \hfill $\Box$ \\[2mm] 
Let $(X,\|\cdot \|)$ be a uniformly convex Banach space with modulus
of convexity $\eta:\NN\to\NN$ as in $(10^*)$ in \cite{Kohlenbach(book)} (p.
4.12). Let $K\subset X$ be a compact subset and $F:K\to \RR$ be
continuous with modulus $\omega$ of uniform continuity. Let $(a_n)$ be a dense
sequence in $K$ and $\alpha:\NN\to\NN$ be a modulus of total boundedness
for $K$ as above, $h:\NN\to\NN^{\NN}$ 
be a sequence of names for $(F(a_n))_{n\in\NN}$
and $g:\NN\to\NN^{\NN}$ be such that for each $n\in\NN,$ $g(n)$
is a name for $\| a_n\|.$ We now assume that
$C:=\text{zer}\;F$ is nonempty closed and
convex and that $\rho$ is
a modulus of regularity for $F$ w.r.t. $\text{zer}\;F.$ \\ It is well-known
that the metric projection of $X$ onto $C$ is well-defined and
single-valued. \\ 
Let $D\in\NN$ be an upper bound on the norm of some zero of $F.$
\begin{theorem}\label{uniformly-convex-case} 
There exists a primitive recursive functional in the sense of Kleene
$\Psi$ which uniformly in functions $\eta,\alpha,\omega,g,h,\rho$
and $D$ satisfying the above requirements
computes the unique
zero \\ $\Psi(\eta,\alpha,\omega,g,h,\rho,D)\in \text{zer}\;F$ of $F$ 
which has minimal norm among all zeros, i.e. the metric projection of
$0$ onto  $\text{zer}\;F.$
\end{theorem}
{\bf Proof:} 
In $\eta$ and $D$ one can easily compute primitive recursively a 
modulus of uniqueness
$\Phi(k):=\Phi(\eta,D,k)\in\NN$ for
the metric projection of $0$ onto $C$, see e.g. \cite[Proposition 17.4]{Kohlenbach(book)}, where we here write this modulus
with $\varepsilon/\delta$ of the form $2^{-k},$ i.e.
\[ \forall k\in\NN\,\forall y_1,y_2\in C\,\left(
    \bigwedge^{2}_{i=1} (\| y_i\|\le \inf\limits_{y\in C} \| y\| +2^{-\Phi(k)})
    \rightarrow
\| y_1-y_2\| \le 2^{-k}\right). \]
Let $k\in\NN$ be given and define
\[ L:=\alpha\left(\max\left\{ \Phi(k+1)+2,\omega(K)\right\}\right), \]
where
\[ K:=\rho\left( \Phi(k+1)+2\right)+2. \]
We consider the set
\[ S_k:=\left\{ n\le L: |(F(a_n))(K)|\le
   2^{-\rho(\Phi(k+1)+2) -1} \right\},\]
   where here again, $(F(a_n))(K)$ is the $2^{-K}$-rational approximation
   of $F(a_n)$ provided by $\widehat{h(a_n)}(K).$ \\ 
$(i):$ $S_k\not=\emptyset.$ Let $z\in \text{zer}\;F.$ Then by the
definition of $\alpha,$ there exists an $n\le L$ such that
\[ \| a_n-z\| <2^{-\omega(K)}. \] By the definition of $\omega$ and
using that $F(z)=0,$ we get $|F(a_n)|<2^{-K}$ and so
\[ |(F(a_n))(K)|<2^{-\rho\left(\Phi(k+1)+2\right)-1},\]
i.e. $n\in S_k.$ \\
$(ii):$ The following implications holds for all $n\in\NN$
\[ n\in S_k\rightarrow \exists z\in \text{zer}\;F \
\left( \| z-a_n\|<2^{-\Phi(k+1)-2} \right), \]
since by the assumption
\[ |(F(a_n))(K)|\le 2^{-\rho(\Phi(k+1)+2)-1} \]
and so
\[ |F(a_n)|<2^{-\rho(\Phi(k+1)+2)} \]
so that we can apply the definition of $\rho.$ \\
Now compute (primitive recursively in $k$ and the other data mentioned
in the theorem) an $n_k\in S_k$ such that for all $m\in S_k$
\[ (+) \ (\| a_n\|)(\Phi(k+1)+2)\le (\| a_m\|)(\Phi(k+1)+2), \]
where $(\| a_i\|)(\Phi(k+1)+2)$ is the $2^{-\Phi(k+1)-2}$-rational 
approximation to $\| a_i\|$ provided by $\widehat{g(i)}(\Phi(k+1)+2).$ \\ 
By $(ii),$ $\exists z\in \text{zer}\; F$ such that
\[ \| a_{n_k}-z\|<2^{-\Phi(k+1)-2}. \]
Let $z_0$ be the unique zero of $F$ with minimal norm. \\
{\bf Claim:} $\| z\|\le \| z_0\| +2^{-\Phi(k+1)}.$ \\ {\bf Proof of Claim:}
Suppose that
$\| z\|> \| z_0\| +2^{-\Phi(k+1)}.$ Then
$\| a_{n_k}\|>\| z_0\|+2^{-\Phi(k+1)} -2^{-\Phi(k+1)-2}$ 
and so
\[ (1) \ (\| a_{n_k}\|))(\Phi(k+1)+2)>\| z_0\|+2^{-\Phi(k+1)}-2^{-\Phi(k+1)-1}
      = \| z_0\| +2^{-\Phi(k+1)-1}. \]
By the definition of $L$ there exists an $l\le L$ with
\[ \| a_l-z_0\| <2^{-\Phi(k+1)-2}\wedge  \| a_l-z_0\| <2^{-\omega(K)}. \]
By the second conjunct we get - reasoning as in $(i)$ - that $l\in S_k.$ \\
The first conjunct implies that
\[ (2) \ (\| a_l\|)(\Phi(k+1)+2)< \| z_0\| +2^{-\Phi(k+1)-1}.\]
$(1)$ and $(2)$ together yield that
\[ (\| a_{n_k}\|)(\Phi(k+1)+2) >(\| a_{l}\|)(\Phi(k+1)+2) \]
which contradicts $(+)$ and so finishes the proof of the claim.
\\ Since $\Phi$ is a modulus of uniqueness, the claim implies that
$\| z-z_0\|\le 2^{-k-1}$ and so (since $\Phi(k)\ge k$ which holds by
the construction of $\Phi$ and which we, anyhow, may assume w.l.o.g. by 
taking $\max\{ \Phi(k),k\}$, otherwise)
$\| a_{n_k} -z_0\|< 2^{-k}.$  This implies that $(n_k)_k$ is a name for
$z_0$ in the sense of \cite{Kohlenbach(book)}.\hfill $\Box$\\[2mm]
The following axiom (formulated in the language of the system
${\cal A}^{\omega}[X,d,W]$ of classical analysis in all finite types
augmented with an abstract type $X$ for a bounded hyperbolic space defined in
\cite{Kohlenbach(05)}) states that for an abstract bounded metric space $(X,d)$
any nonexpansive
function
$F:X\to \RR$ which has a zero possesses a modulus of regularity 
\[ \mbox{\rm (NE-Reg)}: \forall F:X\to\RR,\,r>0 \, (F \ \mbox{n.e.}\,\wedge\exists z\in \text{zer}\; F\rightarrow \exists \rho:\NN\to\NN (\rho \ \mbox{\rm mreg}\ \text{zer}\;F,r)). \]
We will show that this classically false (in the specific sense
of Proposition \ref{counterexample} below) axiom can - differently from the
also classically false nonstandard principle $\exists$-UB$^X$ from \cite{Kohlenbach(06)}
(see also \cite{Kohlenbach(book)}) which is admissible
in logical metatheorems of proof mining - not be added to ${\cal A}^{\omega}[X,d,W]$ to obtain a formal system which admits the extraction of classically
correct uniform bounds for (essentially) $\forall\exists$-theorems.\footnote{So, in particular, (NE-Reg) cannot be derived from a combination of
  $\exists$-UB$^X$ with comprehension over natural numbers, e.g. by an
  adaptation to the abstract type $X$ 
  of how the generalized uniform boundedness principles
  $\Pi^0_k$-UB$^{-}\res$ studied in \cite{Kohlenbach(98)} are formed.}
The argument is already implicit in the proof of \cite{Kohlenbach(CHBC)}[Theorem 3] but
we include it here for completeness.
\begin{proposition}\label{counterexample}
  In ${\cal A}^{\omega}[X,d,W]+\mbox{\rm (NE-Reg)}$ one can
  prove
a sentence of the form
\[ A:\equiv \forall g\in\NN^{\NN}, k\in\NN, \, x^X,p^X,T^{X\to X}\,(T \ \mbox{nonexpansive} \ \to \exists n\in\NN\,A_{\exists}(g,k,x,p,T)), \] where $A_{\exists}(g,k,x,p,T)$ is
a provably extensional $\Sigma^0_1$-formula, 
such that $A$ is not true in all bounded hyperbolic spaces $(X,d,W)$
(in fact not even in all closed bounded convex subsets of $l_2$).
\end{proposition}
{\bf Proof:}
Consider the sentence (abbreviating `nonexpansive' by `n.e.')
\[ A:\equiv\forall  g\in\NN^{\NN}, k\in\NN, \,x^X,p^X, T^{X\to X}\,\left( T \ \mbox{n.e.} \,\wedge\, p=_XTp\to \exists n\in\NN\ d_X(x_n,x_{g(n)})<_{\RR} 2^{-k}\right), \]
where $(x_n)$ is defined as $x_0:=x, \ x_{n+1}:=W_X(x_n,Tx_n,1/2),$
which is equivalent to a sentence of the required logical form noticing that
$=_X\in\Pi^0_1$ and $<_{\RR}\in\Sigma^0_1$.
\\ ${\cal A}^{\omega}[X,d,W]$ proves that the Krasnoselski iteration of $T$
is asymptotically regular, i.e. that $d(x_n,Tx_n)\to 0$ (see 
\cite{Kohlenbach/Leustean,Kohlenbach(book)}).
Since $(x_n)$ is Fej\'er monotone w.r.t.
$\text{zer}\;F$ where $F(x):=d(x,T(x)),$ the assumption $F(p)=_{\RR} 0$
by (NE-Reg) yields a
modulus of regularity w.r.t. $\text{zer}\;F$ and $\overline{B}(p,b)$ where $b^{\NN}$
is the constant witnessing the ($b$-)boundedness of $X$ in
${\cal A}^{\omega}[X,d,W].$ Hence by \cite[Theorem 4.1]{KohlenbachLopezNicolae(2019)}, $(x_n)$ is - reasoning in  ${\cal A}^{\omega}[X,d,W]+\mbox{(NE-Reg)}$ - 
a Cauchy sequence and so for all $k\in\NN,\,g:\NN\to\NN$ 
\[  \exists n\in\NN\ d_X(x_n,x_{g(n)})<2^{-k}. \]
Thus in total we get the provability of $A.$ However,
by \cite{GenelLindenstrauss}, 
there exists a bounded closed and convex subset $C\subset l_2$ and a
nonexpansive selfmapping $T:C\to C$ possessing a fixed point in $C$ and
a point $x_0 \in C$ 
such that $(x_n)$ does not strongly converge. Hence $(x_n)$ is not a
Cauchy sequence and so $\forall k\in\NN\,\forall g\in \NN^{\NN}\,
\exists n\in\NN\ (d_X(x_n,x_{g(n)})<2^{-k})$ does
not hold in this example, where $X:=C$ and $W(x,y,\lambda):=(1-\lambda)x
+\lambda y.$ 
\hfill $\Box$ 
\\[2mm]
In the following we use the formal definition of the binary (`weak')
K\H{o}nig's lemma (in a language with function variables)
as given in \cite{Troelstra(74)} (see also 
\cite{Troelstra(77)}; here $*,\overline{b}x,
lth(n)$ refer to the primitive recursive coding of finite sequences from
\cite{Troelstra(73)}):

\begin{definition}[\cite{Troelstra(74)}] \hspace{8.5cm}  \\
$Tf:\equiv\aquant n^0 ,m^0 (f(n*m)=_0 0\rightarrow fn=_0 0)\wedge
\aquant n^0 ,x^0 (f(n*\langle x\rangle )=_0 0\rightarrow x\le_0 1)$ \\ (i.e.
$T(f)$ asserts that $f$ represents (the characteristic function of)
a binary tree) \\[1mm] 
$T^{\infty}(f):\equiv T(f)\wedge\aquant x^0 \equant n^0 (lth(n)
=x\wedge fn= 0),$ \\ (i.e. $T^{\infty}(f)$ expresses that $f$ represents an
infinite binary tree), \\[1mm] 
{\rm WKL} $:\equiv\aquant f^1 \big(
T^{\infty}(f)\rightarrow\equant b^1\aquant x^0
(f(\overline{b}x)=0)\big)$.
\end{definition}
\begin{definition}\label{regularity-tree}
We say that $\rho:\NN\to\NN$ is a modulus of regularity for an infinite
binary tree (represented by) $f$ w.r.t. infinite paths through $f$ if
\[ \forall h\le_1 1\,\forall k^0 \, \left( f(\overline{h}(\rho(k)))=0\rightarrow
\exists b\le_11\,( \forall x^0 (f(\overline{b}x)=0) \wedge \overline{h}k=
\overline{b}k)\right). \]
\end{definition}
Definition \ref{regularity-tree} can be seen as a special case of
Definition \ref{regularity-zero} for $X:=2^{\NN}$ with the Baire space
metric
\[ d(f,g):=\left\{ \ba{l} 2^{-\mbox{min} \,k [f(k)\not= g(k)]-1}, \
      \mbox{if $f\not= g$} \\
      0, \ \mbox{otherwise}\ea \right. \]
and \[ F:2^{\NN}\to \RR, \ h\mapsto \sum\limits^\infty_{i=0}
\chi(h,i)\cdot 2^{-i},\]
where
\[\chi(h,i)=\left\{\ba{l} 0, \ \mbox{if} \ f(\overline{h}i)=0, \\
1, \ \mbox{otherwise}. \ea \right. \]
Here we use that $d(f,g)<2^{-k}$ iff $\overline{f}(k)=\overline{g}(k)$ and
that
\[ |F(h)|<2^{-n}\to f(\overline{h}n)=0\to |F(h)|\le 2^{-n}.\]
More precisely, if $\rho$ satisfies Definition \ref{regularity-tree}, then
it also is a modulus of regularity in the sense of \ref{regularity-zero}
for zer\,$F$ and -
conversely - if $\rho$ is modulus in the sense of
Definition \ref{regularity-zero}, then
$\tilde{\rho}(k):=\rho(k)+1$ satisfies Definition \ref{regularity-tree}.
\\[1mm]
As shown in \cite{Kohlenbach(moduli)}, the existence of a modulus of regularity
for continuous functions on compact spaces is equivalent to arithmetical
comprehension ACA$_0$ (whereas the existence of the
$\forall \varepsilon\exists \delta$-version is equivalent to WKL). The reason
for this difference in strength is that already the $\forall\varepsilon\exists\delta$-regularity implies $\Sigma^0_1$-LEM (see Theorem
\ref{regularity-LEM}.2) below) which then, when strengthened into
a modulus, becomes $\Sigma^0_1$-comprehension and so - by interation -
arithmetical comprehension. \\ Using
arithmetical comprehension, one can construct the leftmost infinite path
in an infinite binary tree. 
The next result gives an explicit transformation
of a modulus of regularity in the sense of Definition \ref{regularity-tree}
into the leftmost branch:
\begin{proposition}\label{leftmostbranch}
There exists a Kalmar elementary functional $\varphi$ of type 2 (more
precisely $\varphi$ is given by a closed term of G$_3$A$^{\omega}$ as
defined in {\rm \cite{Kohlenbach(95A)}}) such that for any given infinite
binary tree $f$ with modulus of regularity $\rho,$ $\varphi(f,\rho)$
is the leftmost infinite path in $f.$
\end{proposition}
{\bf Proof:} Define $\varphi(f,\rho)(k):=k$-th component $(\sigma)_k$ of
the leftmost finite branch $\sigma$ of length $\rho(k+1)$ in $f$ which
can be searched for by exponentially bounded search.
We may assume that $k\le\rho(k)$ for all
$k.$ Let $k$ be fixed and consider the leftmost finite branch
$\sigma$ of length $\rho(k+1)$ in $f.$ Then the infinite
sequence $\sigma*0^1$ (defined as continuing $\sigma$ by $0$'s) satisfies
$f(\overline{(\sigma*0)}(\rho(k+1))=f(\sigma)=0.$ By the definition of $\rho$ we
get the existence of an infinite path $b\le 1$ with
\[ \overline{b}(k+1)=\overline{(\sigma*0)}(k+1). \]
Let $\tilde{b}\le 1$ be the leftmost infinite path in $f.$
\\ We show that $\forall i\le k\,(\tilde{b}(i)=b(i))$ (and so,
in particular, $\tilde{b}(k)=b(k)=(\sigma)_k$): suppose that \\ 
$\exists i\le k \, \left( \tilde{b}(i)\not=b(i)\right)$ and let $i_0$
be the smallest such $i.$ \\
Case 1: $\tilde{b}(i_0)<b(i_0)=(\sigma)_{i_0}.$ Then $\overline{\tilde{b}}(\rho(k+1))$
would be a finite branch in $f$ of length $\rho(k+1)$ which is more to
the left than $\sigma$ (since $\tilde{b}(i)=b(i)=(\sigma)_i$ for all
$i<i_0$) which contradicts the
definition of $\sigma.$ \\
Case 2: $b(i_0)<\tilde{b}(i_0).$ Then $b$ would be a more to the left
infinite path in $f$ than $\tilde{b}$ contradicting the definition of
$\tilde{b}.$ \hfill $\Box$
\\[1mm]
As shown in \cite[Theorem 4.1]{KohlenbachLopezNicolae(2019)}, a modulus of
regularity always yields (given an approximate solution bound) a
rate of convergence for Fej\'er monotone
algorithms computing approximate solutions. The algorithm in the proof
of Proposition \ref{leftmostbranch}, in fact, is an instance of this:
\begin{proposition}
  Let for a given infinite binary tree $f,$ $(x_k)$ be defined
  as follows: $x_k:=\sigma_k*0^1,$ where $\sigma_k$ is the leftmost finite branch of length
  $k$ in $f.$ Then $(x_k)$ is Fej\'er monotone w.r.t.
  the set $S$ of infinite paths through $f.$
\end{proposition}
{\bf Proof:}
Let $b\in S.$ We have to show that
\[ \forall k \in\NN \, (d(x_{k+1},b)\le d(x_k,b)). \] This in turn follows from
\\[1mm]
{\bf Claim:} \ $\forall k,m\in\NN\, (\overline{x}_km=\overline{b} m
  \to \overline{x}_{k+1} m=\overline{b} m).$\\[1mm]
{\bf Proof of Claim:} let $k,m\in\NN$ be fixed and assume that
$\overline{x}_km=\overline{b} m$. Suppose that for some $i<m$ we would
have that $x_{k+1}(i)\not= b(i)$ and let $i_0$ be the least such $i$.
Note that $\forall j<i_0\, (x_k(j)=b(j)=x_{k+1}(j)).$\\ 
Case 1: $i_0\le k.$ Since $\sigma_{k+1}$ is the leftmost branch of length
$k+1$ while $\sigma_k$ is leftmost of length $k,$ it follows (using that
$x_k(k)=0$) that
$x_{k}$ is to the left of $x_{k+1}$
so that $b(i_0)=x_k(i_0)\le x_{k+1}(i_0).$
Since $\overline{b}(k+1)$ is some finite branch of length $k+1,$
$\sigma_{k+1}$ is left of it and so $x_{k+1}(i_0)\le b(i_0).$ Hence in total
$x_{k+1}(i_0)=b(i_0)$ which is a contradiction to the definition of $i_0.$\\
Case 2: $i_0>k.$ Then by definition $x_{k+1}(i_0)=0=x_k(i_0)=b(i_0)$ which
again contradicts the definition of $i_0.$
\hfill $\Box$
\begin{definition}
\begin{enumerate}
\item
The $\Sigma^0_1$-law-of-excluded-middle principle (with function parameters) 
is defined as\footnote{For a proof-theoretic study of this principle
  (as a first-order principle without function variables) and its computational 
interpretation see
  \cite{Akama} and \cite{AschieriBerardi} respectively.}
\[\Sigma^0_1\mbox{\rm -LEM}: \ \forall f\, \left( \forall n^0\,(f(n)=0)\vee 
\exists n^0\,(f(n)\not=0)\right). \]
\item
The principle of binary choice for $\Pi^0_1$-formulas is defined as 
\[ \hspace*{-1cm}\Pi^0_1\mbox{\rm -AC}_{\le 1}: \ \forall f\,\left( \forall n^0\,\exists m\le_0 1\,\forall k^0\, (f(n,m,k)=0)
\rightarrow \exists g\le_1 \lambda x.1\,\forall n,k\, (f(n,g(n),k)=0)\right). \]
\end{enumerate}
\end{definition}
In the following, EL denotes the system of elementary intuitionistic analysis
from \cite{Troelstra(73)}.

\begin{theorem}\label{regularity-LEM}
\begin{enumerate}
\item
  {\rm EL$+\Sigma^0_1$-LEM$+\Pi^0_1$-AC$_{\le 1}$} proves that for every compact
  metric space $X=\widehat{A}$ any continuous mapping $F:X\to\RR$ having
  a zero is regular w.r.t. $\text{zer}\;F.$
\item
  Already for Lipschitz continuous functions $F:[0,1]\to\RR$ with
  $zer\,F\not=0$, the regularity of $F$ w.r.t. $zer\, F$ implies
  {\rm $\Sigma^0_1$-LEM over EL}.
\end{enumerate}
\end{theorem}
{\bf Proof:} 1) Inspection of the proof of
\cite[Theorem 4.2(1)]{Kohlenbach(moduli)} (see also \cite[Remark 4.3]{Kohlenbach(moduli)}) shows that the claim can be established with induction,
$\Sigma^0_1$-LEM and WKL. WKL in turn is provable in
EL$+$LLPO$+\Pi^0_1-$AC$^{0,0}_{\le 1}$ as follows from the proof of
\cite[Theorem 3]{Kohlenbach(01)}. LLPO trivially follows from
$\Pi^0_1$-LEM and hence - a fortiori - from $\Sigma^0_1$-LEM.
In total the claim of the theorem follows. \\
2) We refine the proof of \cite[Theorem 4.4(2)]{Kohlenbach(moduli)}, which
classically shows that the existence of a modulus of regularity implies
the convergence of bounded monotone sequences of rationals in $[0,1]$
(and hence arithmetical
comprehension ACA$_0$ by \cite[Theorem III.2.2]{Simpson}), to get that intuitionistically the $\varepsilon/\delta$-form
of regularity implies the Cauchy property of bounded monotone sequences
which is known to imply $\Sigma^0_1$-LEM over EL
(see \cite[Theorem 2.(ii)]{Toftdal},
which in turn refers to \cite[5.4.4]{TroelstravanDalen}, and the more recent
\cite{IshiharaNemoto}). Let $(a_n)$ be a nondecreasing sequence of rational
numbers in $[0,1]$ and $f,T:[0,1]\to [0,1]$ be nonexpansive as defined
in the proof of \cite[Theorem 4.4(2)]{Kohlenbach(moduli)} and $x_n:=T^n0.$
Now suppose that the Lipschitz-2 function $F:[0,1]\to \RR, \ F(x):=
|x-Tx|$ is $\forall \varepsilon\exists\delta$-regular and note that $1\in \text{zer}\;F
\not=\emptyset.$ Since $(x_n)$ is Fej\'er monotone w.r.t. $\text{zer}\;F=
Fix(T)$ and asymptotically regular, i.e. $|x_n-Tx_n|\to 0,$ and, in fact, 
with rate of convergence
$n+3$ (see the proof of \cite[Theorem 4.4(2)]{Kohlenbach(moduli)}) 
it follows that $(x_n)$ is a Cauchy sequence: let $k\in\NN$ be fixed
and $n\in\NN$ 
by the $\forall\varepsilon\exists\delta$-regularity be such that for all 
\[ \forall x\in [0,1] \,\left( |x-Tx|<2^{-n}\to \exists p\in Fix(T) \, (|p-x|<2^{-k-1})\right). \]
Then $|x_{n+3}-p|<2^{-k-1}$ for some $p\in Fix(T)$ and so by the Fej\'er
monotonicity of $(x_n)$
\[ \forall m\ge n+3 \, (|x_m-p|<2^{-k-1})\] which implies - as $k$ was
arbitrary - that 
\[ (*)\ \forall k\in\NN \,\exists n\,\forall m,\tilde{m}\ge n\, (|x_m-x_{\tilde{m}}|<2^{-k}). \]
Let $k\in\NN$ be fixed again and $n_k$ be such that $(*)$ holds for $k.$
For $C\in\NN,$ define $n_C:=\max\{ n_k,k+C+3\}.$
Then
\[ |x_{n_C}-Tx_{n_C}|<2^{-k-C} \] and
so - by $T(x)=\frac{1}{2}(x+f(x))$ -
\[ |x_{n_C} -f(x_{n_C})|<2^{-k-C+1} \] which in turn - by the $f_n$-definition
used to define $f$ - implies that
\[ \forall l\le C \, (a_l< x_{n_C}+2^{-k}) \] and so - since $n_C\ge n_k$ - 
\[ \forall l\le C\, (a_l< x_{n_k}+2^{-k+1}). \]
Since $C\in\NN$ was arbitrary, we get
\[ (**)\ \forall l\in\NN \, (a_l< x_{n_k}+2^{-k+1}). \]
Let $l_k\in\NN$ be so large that $2^{-l_k}\cdot n_k<2^{-k}.$ 
Then - reasoning as in \cite{Kohlenbach(moduli)}, p.383, lines 3-7 - 
\[ (***)\ a_{l_k}\ge x_{n_k}-2^{-k}. \]
Indeed,
assume that $a_{l_k}<x_{n_k}-2^{-k}$ and so - since $(a_n)$ is nondecreasing
\[ \forall l\le l_k\, (a_l<x_{n_k}-2^{-k}).\]
By induction on $n$ we show that
\[ (+) \ \forall n\in\NN \, (x_n\le x_{n_k}-2^{-k}+n\cdot 2^{-l_k}): \]
$x_0=0\le a_0\le x_{n_k}-2^{-k}.$ For the induction step (using that
$x_n,a_l\le 1$):
\[ \ba{l} x_{n+1}=\frac{1}{2}\left( x_n
    +\sum\limits^{\infty}_{l=0} 2^{-l-1}\max\{ x_n,a_l\}\right) \\[3mm]
  \stackrel{\text{I.H., assumption}}{\le}
  \frac{1}{2} \left( x_{n_k}-2^{-k}+n\cdot 2^{-l_k} +
    (x_{n_k}-2^{-k}+n\cdot 2^{-l_k})\sum\limits^{\infty}_{l=0}2^{-l-1} +
    \sum\limits^{\infty}_{l=l_k+1} 2^{-l-1}\right)\\[3mm] \hspace*{1cm}=x_{n_k}-2^{-k}+n\cdot
  2^{-l_k}+ \frac{1}{2} 2^{-l_k-1}< x_{n_k}-2^{-k}+(n+1)\cdot 2^{-l_k}. \ea \]
$(+)$ applied to $n:=n_k$ yields that
\[ x_{n_k}\le x_{n_k}-2^{-k}+n_k\cdot 2^{-l_k}< x_{n_k} \]
which is a contradiction and, therefore, establishes $(***).$ \\
So by $(**)$ and $(***)$ together
(using again that $(a_n)$ is nondecreasing) we have shown
that 
\[ \forall l\ge l_k \, \left( a_l\in [x_{n_k}-2^{-k},x_{n_k}+2^{-k+1}]\right)  \]
which yields that $(a_n)$ is a Cauchy sequence. \hfill $\Box$
\\[1mm]
As shown in \cite{Kohlenbach(moduli)}, with classical logic (and WKL),
$\Sigma^0_1$-IA suffices to prove the $\varepsilon/\delta$-regularity 
in the compact case. However, then $\Sigma^0_2$-DNE
\[ \forall f\,\left(\neg\neg \exists n\forall k \,(f(n,k)=0)\to
\exists n\forall k\,(f(n,k)=0)\right)
\] seems to be needed for the proof. To weaken the latter principle to
$\Sigma^0_1$-LEM, one apparently needs a somewhat stronger induction
in order to establish the principle of bounded $\Sigma^0_1$-comprehension
\[ \Sigma^0_1\text{-BCA}: \ \forall f, k\,\exists \sigma\,
\left( lth(\sigma)=k\wedge \forall i<k \,((\sigma)_i=0\leftrightarrow
\exists n \, (f(i,n)=0))\right) \]
which was introduced under the name of AS$^{\Sigma}_1$ in \cite{Parsons(70)}
and studied with the name above in \cite{IshiharaNemoto}.
The situation is, therefore, analogous to that of the Cauchy property
of bounded monotone sequences in $[0,1]$ studied under the name of
PCM$_{ar}$ in \cite{Kohlenbach(PRA)} and \cite{Toftdal} and - under
the name of MCT$^-$ in \cite{IshiharaNemoto}: $\Sigma^0_1$-BCA, which implies
both $\Sigma^0_1$-LEM and $\Sigma^0_1$-IA over EL restricted to
quantifier-free induction (and even weaker systems;
see \cite{IshiharaNemoto}), and which is implied by $\Sigma^0_1$-LEM and IA,
suffices for the proof Theorem \ref{regularity-LEM}.1) but it is open whether
here IA can be weakened to $\Sigma^0_1$-IA. This, of course, does not come
as a surprise since - as the proof of Theorem \ref{regularity-LEM}.2) shows -
metric regularity intuitionistically implies PCM$_{ar}$ (the case of
sequences of reals can easily be reduced to the one for rationals using
rational approximations as in the proof of \cite[Proposition 5.2(1)]{Kohlenbach(PRA)}).
\\[1mm]
{\bf Acknowledgment:} The author is grateful to Pedro Pinto and
Nicholas Pischke for comments
which improved the presentation in this paper. \\ The author was supported by the German Science Foundation (DFG
Project KO 1737/6-2).
\nocite{*}
\bibliographystyle{eptcs}
\bibliography{content-modulus-regularity-eptcs-bib}
\end{document}